\documentclass[a4paper]{article}
\pdfoutput=1
 
 \usepackage[margin=1.25in]{geometry}
\usepackage{verbatim}
\usepackage{amssymb,enumerate}
\usepackage{graphicx}
\usepackage{amsbsy}
\usepackage{authblk}
\usepackage{amsmath,amssymb,color}
\usepackage{graphicx}
\usepackage{amsmath,amsfonts}
\usepackage{enumerate}
\usepackage{verbatim,amssymb,amsfonts,pifont,paralist}

\newtheorem{theorem}{Theorem}

\newenvironment{proof}{\paragraph{Proof:}}{\hfill$\square$}

\title{On String Contact Representations in 3D}

\author{Debajyoti Mondal\thanks{Work of the author is supported in part by the Natural Sciences and Engineering Research Council of Canada
 (NSERC).} } 

\affil{Cheriton School of Computer Science, University of Waterloo, Canada\\
  \texttt{dmondal@uwaterloo.ca}}

\newcommand{\C}{\mathcal{C}}

\begin{document}

\maketitle

\begin{abstract}

An \emph{axis-aligned string}  is  a simple polygonal path, where each line segment is parallel to an axis in $\mathbb{R}^3$. Given a graph $G$, a \emph{string contact representation} $\Psi$ of $G$ maps the vertices of $G$ to interior disjoint  axis-aligned strings, where no three strings meet at a point, and two strings share a common point if and only if their corresponding vertices are adjacent in $G$. The \emph{complexity of $\Psi$} is the minimum integer $r$ such that every string in $\Psi$ is a \emph{$B_r$-string}, i.e., a string with  at most $r$ bends.  
 While a result of Duncan et al. implies that every graph
 $G$ with maximum degree 4 has a string contact representation using $B_4$-strings, we examine constraints on $G$ 
 that allow string contact representations with complexity 3, 2 or 1. We prove that 
 if $G$ is Hamiltonian and  triangle-free, then $G$ admits a contact representation where
 all the strings but one are $B_3$-strings.
 If $G$ is 3-regular and bipartite, then $G$ admits a contact representation with string complexity 2, and if  
 we further restrict $G$ to be Hamiltonian, then $G$ has a contact representation, where all the strings
 but one are $B_1$-strings  (i.e., $L$-shapes).  Finally, we prove some complementary lower bounds on
 the complexity of string contact representations. 
\end{abstract}

\section{Introduction}
A \emph{contact system} of a geometric shape $\xi$ (e.g., line segment, rectangle, etc.) is an arrangement of a set of geometric objects of shape $\xi$, where two objects may touch,  but cannot cross each other. Representing graphs as a contact system of geometric objects is an active area of research in graph drawing.  Besides the intrinsic theoretical interest, such representations find application in many  applied fields such as cartography, VLSI floor-planning, and data visualization.  In this paper  we examine contact systems of \emph{axis-aligned strings}, where each object is a simple polygonal path with axis-aligned straight line segments. No two strings are allowed to cross, i.e., any shared point must be an end point of one of these strings.  

\begin{figure}[pt]
\centering
\includegraphics[width=\textwidth]{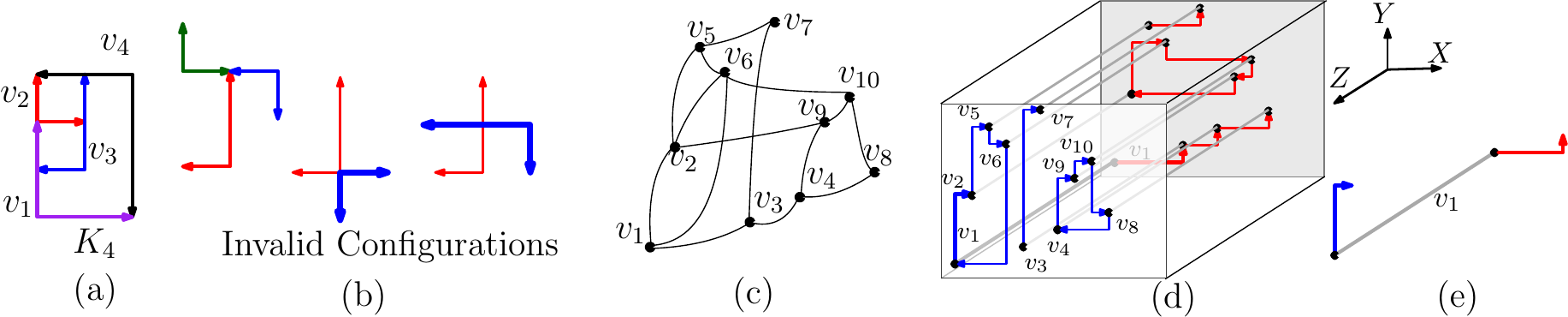}
\caption{(a)   A string contact representations of $K_4$. We use arrows to mark the two ends of each string. (b)  Some invalid configurations. (c) A graph $G$. (d) A string contact representation of $G$. (e) The string corresponding to vertex $v_1$.}
\label{fig:pcr}
\end{figure}

A \emph{string contact representation} of a graph $G$ is a contact system $\Psi$ of  axis-aligned strings in $\mathbb{R}^3$, where each vertex is represented as a distinct string in $\Psi$, no three strings meet at a point,  and two strings touch if and only if the corresponding vertices are adjacent in $G$, e.g., see Fig.~\ref{fig:pcr}. The reason we forbid more than two strings to meet at a point is to avoid degenerate cycles. By a $B_k$-string we denote a string with at most $k$ bends. The \emph{complexity of $\Psi$} is the minimum integer $r$ such that every string in $\Psi$ is a $B_r$-string. 
 We discuss the related research  in two broad categories, first in 2D and then in 3D. 

\textbf{Two Dimensions:} Contact representations date back to the 1930's, when Koebe~\cite{Koebe} proved that every planar graph can be represented as a contact system of circles in the Euclidean plane. A rich body of literature examines contact representation of planar graphs in $\mathbb{R}^2$ using axis-aligned rectangles~\cite{BhaskerS88,KantH93,KozminskiK85} and polygons of bounded size~\cite{AlamEGKP14,AlamBFGKK13,DuncanGHKK12}. In 1994, de Fraysseix et al.~\cite{fraysseix_T} proved that every planar graph admits a triangle contact representation, and showed how to transform it into a contact system of $T$- or $Y$-shaped objects. Subsequent studies involve constructing contact representations with simpler shapes such as axis-aligned segments ($B_0$-strings), and  axis-aligned $L$ shapes ($B_1$-strings). Not all planar graphs can be represented using these shapes. Planar bipartite graphs~\cite{CzyzowiczKU98} and planar Laman graphs~\cite{KobourovUV13} can be represented using $B_0$-strings  and $B_1$-strings, respectively. Recently, Aerts and Felsner~\cite{f2015} examined contact representations of planar graphs using general strings. 
 \emph{Intersection representation} (or, \emph{$B_kVPG$-representation}, where all the strings are $B_k$-strings) is another related concept, where the strings  are allowed to cross. Graphs with $B_rVPG$-representations  
  do not necessarily have contact representations with $B_r$-strings (e.g., $K_{3,3}$ with $r=0$ has a $B_rVPG$-representation, but does not have a contract representation 
 with $B_0$-strings,  as shown in Section~\ref{lb}). 
  We refer the reader to~\cite{DBLP:journals/dcg/ChalopinGO10,DBLP:journals/jgaa/ChaplickU13,BiedlD15,FelsnerKMU16} for further background on $B_kVPG$-representation of planar and non-planar graphs.



\textbf{Three Dimensions:} Contact representation in three dimensions has been examined using axis-aligned boxes~\cite{AdigaC12,Thomassen86} and polyhedra~\cite{AlamEKPTU15}. In the context of geometric thickness, Duncan et al.~\cite{DuncanEK04} proved that the edges of every graph $G=(V,E)$ with maximum degree 4 can be  partitioned into two planar graphs $G_1=(V,E_1)$ and $G_2=(V,E_2)$, each   consists of a set of paths and cycles. They showed that $G_1$ and $G_2$ can be drawn simultaneously on two planar layers with vertices at the same location and edges as $1$-bend polygonal paths. Such a drawing can easily  be transformed into a contact representation of $B_4$-strings (e.g., see Figs.~\ref{fig:pcr}(c)--(e), details are in Appendix~\ref{sec:4-string}), and hence, every graph with maximum degree 4 has a string contact representation with complexity 4.

Not much is known about string contact representations with low complexity strings in $\mathbb{R}^3$. The challenge is vivid even in extremely restricted scenarios: Given a graph along with a label $East$, $West$, $North$, $South$, $Up$, or $Down$,  the problem of computing a no-bend orthogonal drawing in $\mathbb{R}^3$  
 respecting the label constraints has lead to significant research outcomes~\cite{DBLP:journals/dcg/BattistaKLLW12,DBLP:conf/gd/GiacomoLP02}, even for apparently simple structures such as paths, cycles, or graphs with at most three cycles. 
 
Orthogonal drawings can sometimes be turned into string contact representations. Consider a graph $G$ that admits an edge orientation such that the outdegree of every vertex is at most two (e.g., a $(2,0)$-sparse graph~\cite{DBLP:journals/corr/AlamE0KPSU15}). String contact representation with bend complexity 14 can easily be computed for such graphs, e.g., see Fig.~\ref{general} in Appendix~\ref{app:general}. Specifically, if $G$ admits a $k$-bend orthogonal drawing, then the drawing can be turned into a string contact representation with  complexity $(2k+1)$ by forming for each vertex, a string that consists of the outward edges.   However, computing orthogonal drawings with low number of bends per edge is a challenging problem~\cite{HandbookOrtho}. 
 To the best of our knowledge, the complexity of deciding whether a graph has an 
 orthogonal drawing in $\mathbb{R}^3$ with one bend per edge is open.

\textbf{Contributions.} 
 We present significant progress in characterizing  graphs (possibly non-planar)  that admit string contact representations in $\mathbb{R}^3$.   
 We  prove that  every Hamiltonian and triangle-free graph $G$ has a contact representation, where all the strings but one are $B_3$-strings. Using a slightly different construction we show  that every bipartite 3-regular graph admits a string contact representation with  complexity 2. Most interestingly, we prove that every 3-regular graph that is Hamiltonian and bipartite  has a contact representation, where all the strings but one are $B_1$-strings (i.e., $L$-shapes). This construction relies on a deep understanding of the  graph structure and the geometry of $L$-contact systems. All proofs are constructive, and  can be carried out in polynomial time. 

In contrast, we prove (by a simple counting argument) that  5-regular graphs  do not have  string contact representations, even with arbitrarily large  complexity. Moreover, the 4-regular graph $K_5$ (resp., the 3-regular graph $K_{3,3}$)  cannot be represented using $B_1$ strings (resp., $B_0$-strings). 


\section{Preliminaries}
\label{pre}
 

 We assume familiarity with basic graph-theoretical notation. A \emph{straight-line drawing} of a graph $G$ is a drawing in $\mathbb{R}^d$, where each vertex of $G$ is mapped to a point, and each edge of $G$ is mapped to a straight line segment between its end vertices. The \emph{geometric thickness} of $G$ is the minimum integer $\theta$ such that $G$ admits a straight-line drawing $\Psi$ in $\mathbb{R}^2$ and a partition of its edges into $\theta$ sets, where no two edges of the same set cross (except possibly at their common end points), e.g., see Fig.~\ref{even}(a). 

Let $\Psi$ be a contact representation of a graph $G$, where all the strings are  \emph{$L$-shapes}, i.e., $B_1$-strings. For any vertex $v$ of $G$, we denote by $L_v$ the $L$-shape corresponding to $v$ in $\Psi$. Let  $a,o,b$ be the polygonal path representing $L_v$. We refer to $o$ as the \emph{joint} of $L_v$, and the line segments $ao$ and $bo$ as the \emph{hands} of $L(v)$. The points $a$ and $b$ are called the \emph{peaks} of $ao$ and $bo$, respectively. By $\Pi_{xy},\Pi_{yz},\Pi_{xz}$ we denote the family of  planes parallel to the $XY, YZ,$ or $XZ$-plane, respectively.   By $\Pi(t)$ we denote the plane $z=t$.  	For $Q\in\{X,Y,Z\}$, a \emph{$(+Q)$-arrow} is a   directed straight line segment, which is aligned to the $Q$-axis and directed to the positive $Q$-axis. Define a $(-Q)$-arrow symmetrically. A  \emph{$Q$-line} (resp., \emph{segment}) is a straight line (resp., segment) parallel to the $Q$-axis. Throughout the paper the terms `horizontal' and `vertical' denote alignment with $X$ and $Y$-axis, respectively.

\begin{figure}[pt]
\centering
\includegraphics[width=.6\textwidth]{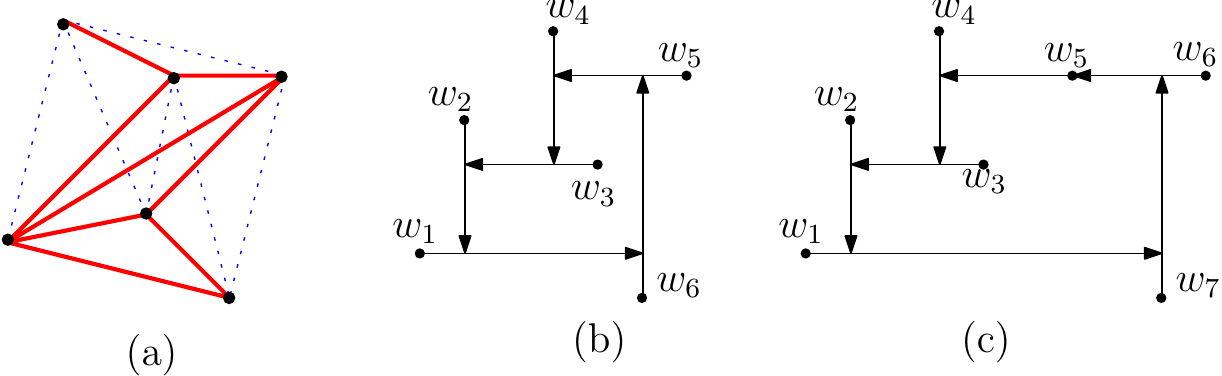}
\caption{(a) A geometric thickness two representation. (b)--(d) Construction of a staircase representation. }
\label{even}
\end{figure}

Let $\mathcal{W}=(w_1,w_2\ldots,w_k)$ be a cycle of length $k\ge 4$. We define a \emph{staircase representation} of $\mathcal{W}$ as a contact system of directed line segments (\emph{arrows}, or degenerate $L$-shapes), 
 as illustrated in Figs.~\ref{even}(b)--(c). If $k$ is even, then the origins of the $L$-shapes are in  \emph{general position}, i.e., no two of them  have the same $x$ or $y$-coordinate. Otherwise, all the origins except  for the two topmost horizontal arrows are in general position. Appendix~\ref{defn} includes a formal definition.

\section{String Contact Representations of Complexity  2 or 3}
\label{string-contact}



\begin{theorem}
\label{4-reg-gu} 
Every  triangle-free Hamiltonian graph $G$ with maximum degree four has a contact representation where all strings but one are  $B_3$-strings. 
\end{theorem}
\begin{proof}[Proof Outline]
Let $C=(v_1,\ldots,v_n)$ be a Hamiltonian cycle of $G$. Let $H$ be the graph obtained after removing the  Hamiltonian edges from $G$. 
 Observe that $H$ is a union of vertex disjoint cycles and paths. We transform each path $P=(w_1\ldots,w_k)$ of $H$ into a cycle by adding a subpath of one or two dummy vertices between
 ($w_1$ and $w_k$) depending on whether $P$ has one or more vertices. 


 Let $Q_1,\ldots, Q_k$ be the cycles in $H$. For each cycle $Q_i$, where $1\le i\le k$, we construct a staircase representation $\Psi_i$ of $Q_i$ on $\Pi(0)$. If $Q_i$ is a cycle with odd number of vertices, then we construct the staircase representation such that the leftmost segment among the topmost horizontal segments corresponds to the vertex with the lowest index in $Q_i$. For example, see the topmost staircase of Fig.~\ref{3-string-short}(a). We then place  the staircase representations diagonally along a line with slope $+1$. We ensure that the horizontal and vertical slabs containing $\Psi_i$ do not intersect $\Psi_j$, where $1\le j(\not=i)\le k$. We  refer to this representation as $\Psi_H$.  
 
Consider now the edges of the Hamiltonian cycle $C=(v_1,\ldots,v_n)$. Note that each vertex $v_j$, where $1\le j\le n$, is represented using an axis-aligned arrow $r_j$ in $\Psi_H$. 
 For each $r_j$,  we construct a $(+Z)$-arrow $r'_j$ of length $j$ that starts at the origin of $r_j$. Consequently, the plane  $\Pi(j)$ intersects only those arrows $r'_q$, where $j\le q\le n$. Let $I_j$  be the set of intersection points on $\Pi(j)$. By construction $\Psi_H$ satisfies the following sparseness property: 
 Any vertical (resp., horizontal) line on $\Pi(j)$ contains at most one point (resp., two points) from $I_j$. For every pair of points $p,q$ that belong to $I_j$ and lie on the same horizontal line, the corresponding vertices are adjacent in $H$, and belong to a distinct cycle with odd number of vertices in $H$. 

\begin{figure}[pt]
\centering
\includegraphics[width=\textwidth]{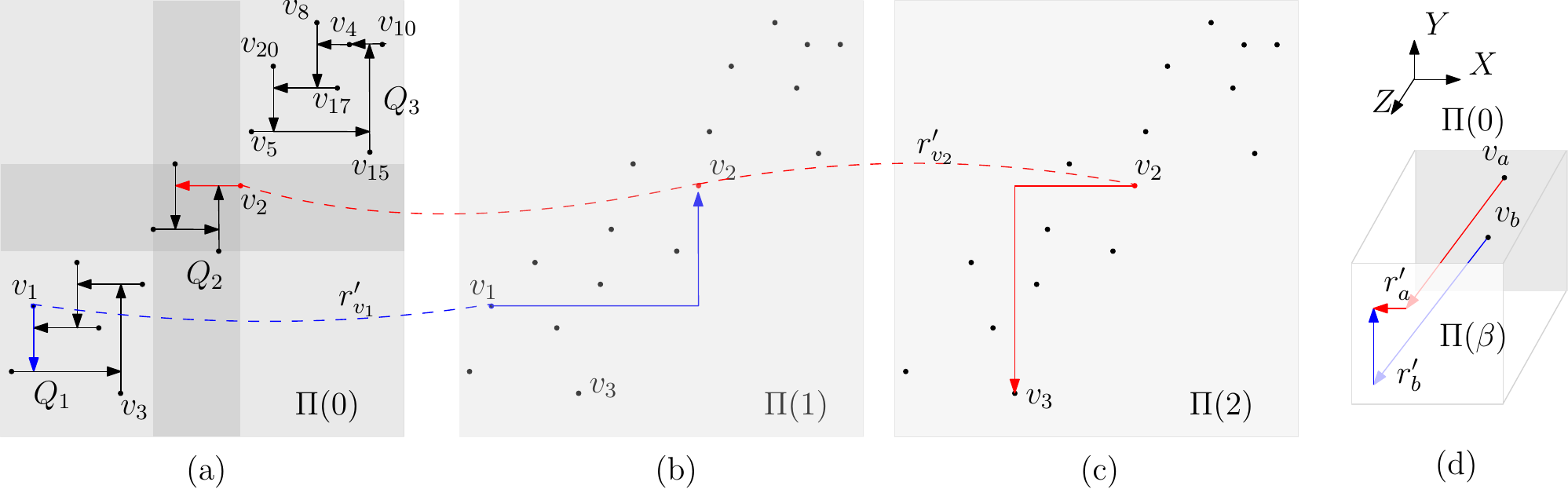}
\caption{ Illustration for the proof of Theorem~\ref{4-reg-gu}: (a) $\Psi_H$, (b)--(c) Extensions of $r'_{v_1}$ and $r'_{v_2}$, where the dashed lines are $(+Z)$-lines. (d) Illustration for Theorem~\ref{3-reg-b}. 
}
\label{3-string-short}
\end{figure}

For each $j$ from $1$ to $n-1$, we realize the edge  $(v_j,v_{j+1})$ by extending $r'_j$ on $\Pi(j)$. Note that it suffices to use two bends to route $r'_j$ to touch $r'_{j+1}$, where one bend is to enter $\Pi(j)$ and the other is to reach $r'_{j+1}$.  Figs.~\ref{3-string-short}(b)--(c) illustrate  the extension of $r'_j$.  We use the sparseness property of $\Psi_H$ to show that  can find such an extension of $r'_j$ without introducing any crossing. Details are in Appendix~\ref{sec:4-string}. Finally, it is straightforward to realize $(v_1,v_n)$ by routing $r'_n$ on $\Pi(n)$ using two bends, and then moving downward to touch $r_1$. Therefore, the string representing $v_n$  is a $B_4$-string.
\end{proof}




\begin{theorem}
\label{3-reg-b} 
Every 3-regular bipartite graph has a string contact representation with  complexity 2. 
\end{theorem}
\begin{proof}
By Hall's condition~\cite{Hall}, $G$ contains a perfect matching   $M\subset E$.  Let $H$ be the graph obtained by removing the edges of $M$ from $G$. Since $H$ is 2-regular, $H$ is a union of disjoint cycles. 
We now construct a contact representation $\Psi_H$ of $H$   in the same way as in the proof of Theorem~\ref{4-reg-gu}.  However, while constructing the $(+Z)$-arrows, we take the matching into consideration. For each edge $(v_a,v_b)\in M$,   we set the length of $r'_a$ and $r'_b$ to $\beta = \min(a,b)$.  Consequently, we can route both $r'_a$ and $r'_b$ to touch each other on $\Pi(\beta)$, e.g., see Fig.~\ref{3-string-short}(d).   Since $G$ is bipartite, $H$ contains only cycles of even number of vertices. Consequently, the origins of the $(+Z)$-arrows are in general position, and hence the extensions of $r'_a$ and $r'_b$ do not create any unnecessary adjacency.
\end{proof}   

\section{$\mathbf{\it L}$-Contact Representations}
\label{sec:lcontact}
The techniques used in Section~\ref{string-contact} inherently require strings with two or more bends. 
 In this section we restrict our attention to contact representations of $B_1$-strings ($L$-shapes). We prove that every 3-regular Hamiltonian bipartite graph  $G$  has a contact representation where all strings but one are  $B_1$-strings. Appendix~\ref{app:last} illustrates a walkthrough example.

\textbf{Technical Details:} 
 Let $P=(v_1,\ldots,v_n)$ be a Hamiltonian path in $G$.  Let $G'$ be the graph obtained after deleting the edge $(v_1,v_n)$ from $G$. We first construct an $L$-contact system for $G'$, and then extend this contact system to compute the representation for $G$.  

Color all the vertices of $G'$, as follows:  Order the vertices from left to right in the order they appear on $P$ (in the increasing  order of indices). For each non-Hamiltonian edge $(v_i,v_j)$, where $j>i+1$, color $v_i$ and $v_j$ with red and blue colors, respectively. Since each vertex is incident to one non-Hamiltonian edge, all the vertices are now colored. This vertex coloring creates \emph{red} and \emph{blue chains} (maximal subpath containing vertices of the same color) on $P$. Let $\C_1,\C_2,\ldots, \C_k$ be all the red chains in $P$ in the left to right order, e.g., see   Fig.~\ref{decomposition}(a). For each $\C_i$, where $1\le i\le k$, there is a blue chain $\C'_i$ that follows $\C_i$. We refer to $(\C_i,\C'_i)$ as a \emph{chain pair}. Let $\C_i$ be the red chain $v_j,\ldots,v_k$, where $1\le j,k \le n$. Since $\C_i$ is maximal, the vertex $v_{j-1}$ and  $v_{k+1}$ (if they exist) are blue vertices. We call $v_{j-1}$ and  $v_{k+1}$ the  \emph{head} and \emph{tail vertex} of $\C_i$, e.g., see  Fig.~\ref{decomposition}(b).  The set $blue(\C_i)$ consists of all blue vertices of $G'$ (following $v_k$ on $P$) that are incident to the vertices of $\C_i$. For example, in Fig.~\ref{decomposition}, $blue(\C_2)$ contains 4 blue vertices. For the $j$th blue vertex $w$ (from left) on $\C'_1$, define  $\alpha(w)$ to be $\lfloor\frac{j}{2}\rfloor+1$, e.g., see  Fig.~\ref{decomposition}(c). For $i>1$, define $\alpha(w)$ (for the $j$th vertex $w$ on $\C'_i$) to be $\delta_{i-1} + \lfloor\frac{j}{2}\rfloor+1$, e.g., see  Fig.~\ref{decomposition}(c), where $\delta_{i-1}$ is the maximum $\alpha$ value (-1, if the maximum is even and unique) in $\C'_{i-1}$.
 %
 Finally, define $G'_i$ to be the graph induced by the edges of $\C_1,\C'_1,\ldots,\C_i,\C'_i$, along with the edges that connect blue vertices of $G'$ to these chains, e.g., see  Fig.~\ref{decomposition}(d). 
 A vertex $v$ is \emph{unsaturated} in $G'_i$, if $v$ has a neighbor in $G'$ that does not belong to $G'_i$. 
 Otherwise, $v$ is a \emph{saturated} vertex.  Every blue vertex $w$ in $G'$ must be incident to a non-Hamiltonian edge $(u,w)$ such that $u$ is red and appears before  $w$ on $P$. We call $u$ the \emph{red parent} of $w$, and $w$ the \emph{blue child} of $u$.

\textbf{Idea:} We construct the $L$-contact representation of $G'$ incrementally, starting from $G'_1$, and then at  the $i$th step, adding the chain pair $(\C_i,\C'_i)$ and the edges that connects $blue(\C_i)$ to $\C_i$.
 In other words, after the $i$th step, we will have an $L$-contact representation $\Psi'_i$ of $G'_i$. 
 For each $i$ from $1$ to $k$, we construct $\Psi'_i$   maintaining some drawing invariants. 
 
 In brief, we will draw the red chain $\C_1$ as a contact representation of arrows
  (degenerate $L$-shapes), where the arrows will be arranged along an $xy$-monotone polygonal path 
  lying on plane $\Pi(1)$, e.g., see Fig.~\ref{Gamma1}(a).
  For each red vertex $v$ and non-Hamiltonian edge
  $(v,w)$, we draw the other hand of $L_v$ as a $(+Z)$-arrow that stops at $\Pi(\alpha(w))$.
  The intuition is that the joint (of $L$-shapes) of the blue vertices   will be drawn on 
  the plane defined by their $\alpha(\cdot)$ values.  
  Since every  blue vertex $w$ in $\C'_1$ and $blue(\C_1)$  has a red parent in $v$  in $\C_1$, 
  we draw  $L_w$ initially as a point (degenerate $L$-shapes) at the peak of $L_v$, e.g., see Fig.~\ref{Gamma1}(c).
  Thus to complete the drawing of $\Psi'_1$, we only need to realize the edges of $\C'_1$,
  which is done by extending the degenerate blue $L$-shapes. The  $\alpha(\cdot)$ values
  will play a crucial role to ensure that the blue $L$-shapes follow
  some increasing $Z$-direction, and thus can be drawn without introducing any unnecessary
   adjacency.
   
 For $i>1$, the there are two key  differences between  $(\C_i,\C'_i)$ and $(C_1,C'_1)$. First, $\C_i$ 
  has a head vertex $v_h$, which is already drawn in $\Psi'_{i-1}$. Second, $\C'_i$ and $blue(\C_i)$
  may contain red parents that do not belong to $\C_i$, and thus already drawn in $\Psi'_{i-1}$.  
  The most favorable scenario would be to construct a drawing of $(\C_i,\C'_i)$ and the
  edges connecting them to  $blue(\C_i)$ independently (following the drawing method of $\Psi'_1$),  
  and then insert it into $\Psi'_{i-1}$ to obtain the drawing $\Psi'_i$. 
  If the red parents of all the vertices in $\C'_i$ and $blue(\C_i)$ belong to $\C_i$, then
  we can easily construct  $\Psi'_i$ using the above idea.  Otherwise,  merging the drawings 
  properly 
 seems challenging. However, using the drawing invariants we can find certain properties in 
  $\Psi'_{i-1}$ that makes such a merging possible. 
   
\textbf{Drawing Details:} For each $i$ from $1$ to $k$, we construct $\Psi'_i$   maintaining the following drawing invariants. 

\begin{figure}[pt]
\centering
\includegraphics[width=.8\textwidth]{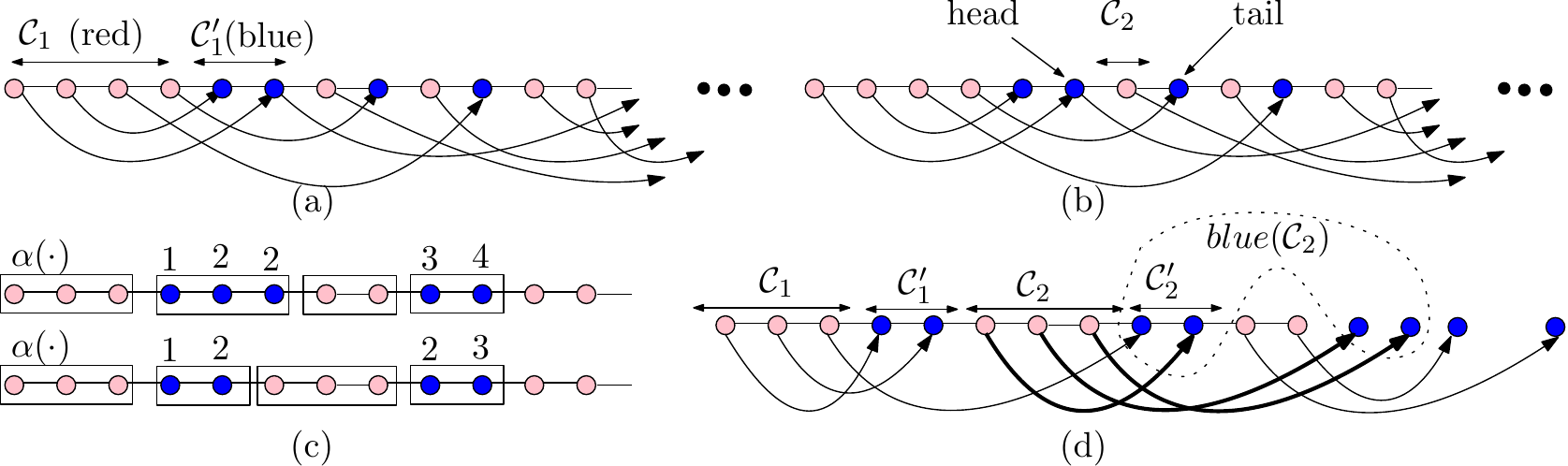}
\caption{(a) $\C_i$, and $\C'_i$. (b) The head and tail of $\C_2$.  Illustration for (c) $\alpha(\cdot)$, where $\delta_{1} = 2$ (top) and $\delta_{1} = 1$  (bottom), respectively. (d) $G'_2$.}
\label{decomposition}
\end{figure}

\smallskip 
\begin{compactenum}
\item [$I_1$.] $\Psi'_i$ is an $L$-contact representation of $G'_i$. 

\item[$I_2$.] Every blue vertex $v$ of degree one in $G'_i$ is drawn as a point
  on $\Pi(\alpha(v))$.  The projection of these points  on $\Pi_{xy}$  are in general position. 

\item [$I_3$.] Let $w_j$ be the $j$th blue vertex on $\C'_i$ (from left to right). 
 If $j$ is odd, 
 then  $\Pi(\alpha(w_{j+1}))$ contains only one 
 point (peak) of $L_{w_j}$, where the rest of $L_{w_j}$ lies below  $\Pi(\alpha(w_{j+1}))$.  
 Otherwise,  $\alpha(w_j) = \alpha(w_{j+1})$, 
 and $\Pi(\alpha(w_{j}))$ 
 contains entire $L_{w_j}$. Moreover, $L_{w_j}$ is non-degenerate, and 
 the $Y$-line through the joint of 
 $w_{j+1}$ must intersect (the extension of) the horizontal hand of $L_{w_j}$.  
 %
\end{compactenum}
\smallskip 

\noindent
For simplicity we do not introduce drawing invariants for the red $L$-shapes.
 Their drawing will be obvious from the context.
  Informally, for a red chain $\C_i$,
  one hand of the corresponding $L$-shapes will be drawn on the plane $\Pi(1)$ (if $i=1$) or 
  on the plane determined by the $\alpha(\cdot)$ value of its head (if $i>1$). 
  The remaining hand (if needed) is drawn as a $(+Z)$-arrow that stops at 
  some plane determined by the $\alpha(\cdot)$ value of its blue child.
  

\subsection{Construction of $\Psi'_1$} 
Let $\C_1$ be the red chain $v_1,v_2,\ldots,v_q$. 
 The tail $v_t$ of $\C_1$, which is blue, is adjacent to
 exactly two vertices of $\C_1$: One is the red vertex $v_q$, and the other is 
 its red parent $v_r$. While constructing $\Psi'_1$,  we first realize 
 the red-red adjacencies, then  the red-blue  (equivalently, blue-red) adjacencies,
 and finally, the blue-blue  adjacencies of $G'_1$. 

\textbf{Red-Red Adjacencies:} Red-red adjacencies correspond to Hamiltonian edges, and thus appear in $\C_1$. 
 To realize these adjacencies, we draw the $L$-shapes of the vertices of $\C_1$ using arrows
 that lie along an $xy$-monotone path on $\Pi(1)$, as illustrated in Fig.~\ref{Gamma1}(a).
 
 In brief, we ensure that the arrows are horizontal and vertical alternatively,
  and all the joints (origins) 	are in general position. We then extend the 
   joint of $L_{v_r}$ such that  the $Y$-line through it  intersects $L_{v_q}$, e.g., see  Fig.~\ref{Gamma1}(b).    All these conditions are straightforward to achieve.



   
\begin{figure}[pt]
\centering
\includegraphics[width=\textwidth]{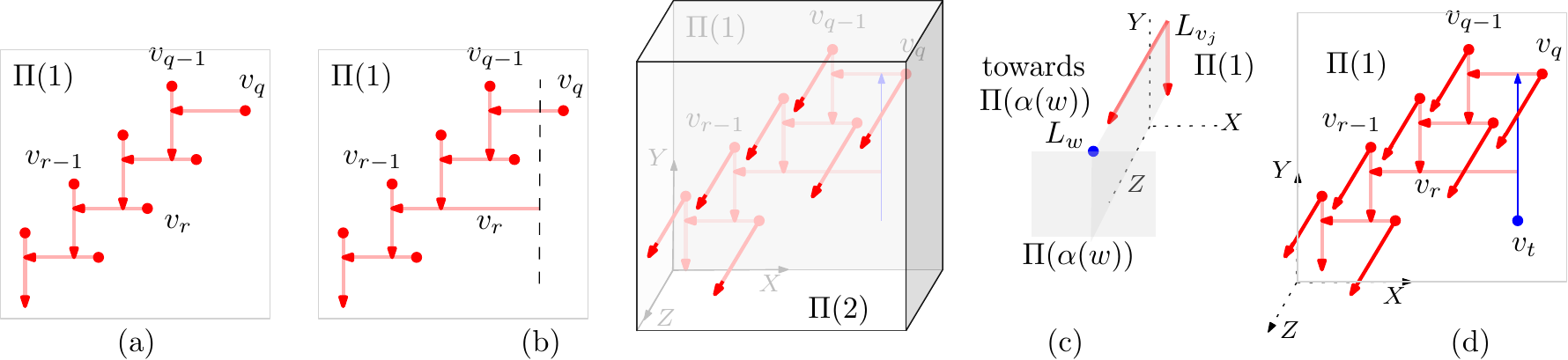}
\caption{Illustration for the drawing of $\Psi'_1$. (a)  Initial drawing. (b) Construction for $L_{v_r}$,  (c)   The $(+Z)$-arrow of $L_{v_j}$. (d) $L_{v_t}$.}
\label{Gamma1}
\end{figure}

\textbf{Red-Blue Adjacencies:}  For each red vertex $v$ (except for $v_r$), we create a $(+Z)$-arrow  (the other 
 hand of $L_v$) that stops at $\Pi(\alpha(w))$, where $w$ is the blue child of $v$. 
 We then draw $w$ as a point at the peak of the arrow, e.g., see  Fig.~\ref{Gamma1}(c). 
 We will refer to such an initial point representation of $w$ as the \emph{initiator of $w$},
 and denote the point as $init(w)$. 
 Although the joint of such a point representation of $w$ coincides with $init(w)$, 
  it is important to note that we may later extend the point representation of $w$ 
 to an arrow or a full $L$-shape, and the joint of the new $L$-representation 
 does not necessarily coincide with $init(w)$. 
 
 The only remaining red-blue adjacencies are $(v_q,v_t)$
  and $(v_r,v_t)$. Recall that the $Y$-line through the joint of $L_{v_r}$
  intersects $L_{v_q}$. Therefore, we can draw a $(+Y)$-arrow representing 
 $L_{v_t}$ that touches both $L_{v_r}$ and $L_{v_q}$, e.g., see  Fig.~\ref{Gamma1}(d).



\textbf{Blue-Blue Adjacencies:} Let $w_1(=v_t),\ldots,w_k$ be the blue chain $\C'_1$. If $\C'_1$ does not include all the blue vertices of $G'$, then let $w_{k+1}$ be the first blue vertex following $w_k$ on $P$. Note that  $w_k$ and $w_{k+1}$ are the head and tail of  $\C_2$, respectively, e.g., see Fig.~\ref{decomposition}(b). On the other hand, if  $\C'_1$ contains all the blue vertices of $G'$, then consider a dummy vertex $w_{k+1}$. 

If $k=1$, then there is no blue-blue adjacency to be realized. We only construct a $(+Z)$-arrow that starts at the   $init(w_1)$ and  stops at $\Pi(\alpha(w_2))$. This satisfies the invariant  $I_3$ (since $w_1$ is at odd position on $\C'_1$). 

\begin{figure}[pt]
\centering
\includegraphics[width=\textwidth]{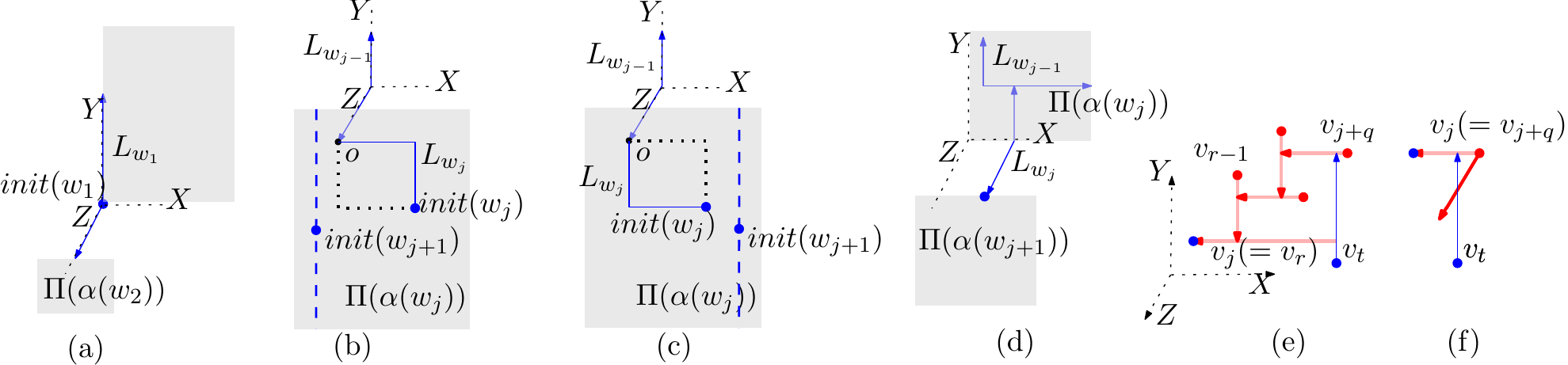}
\caption{Illustration for blue-blue adjacencies.}
\label{invariants}
\end{figure}

If $k>1$, then we modify  $L_{w_j}$, where $1\le j\le k$,  to realize the blue-blue adjacencies.  Observe that each  $L_{w_j}$, except $L_{w_1}$, is  currently represented as a point on $\Pi(\alpha(w_j))$. 
 We first construct a $(+Z)$-arrow which starts at $init(w_1)$ and stops at $\Pi(\alpha(w_2))$,
 i.e., $L_{w_1}$ satisfies the invariant  $I_3$, e.g., see Fig.~\ref{invariants}(a). 
 Consider now the modification for $L_{w_j}$, where $j>2$. Assume that 
 the $L$-shapes $L_{w_1}, \ldots,L_{w_{j-1}}$ already satisfy Invariant $I_3$.   

 If $j$ is even, then $(j-1)$ is odd and by definition of $\alpha(\cdot)$,
 $\alpha(w_{j-1})<\alpha(w_j)$. 
 By Invariant $I_3$, $L_{w_{j-1}}$ has only one point (a peak)
 $o$ on $\Pi(\alpha(w_j))$.
 We now have two options to create $L_{w_j}$ connecting $o$
 and $init(w_j)$. One of these two options  would satisfy Invariant $I_3$, e.g., see Figs.~\ref{invariants}(b)--(c).

If $j$ is odd, then by the definition of $\alpha(\cdot)$, we have $\alpha(w_{j-1}) =  \alpha(w_j)$. Since $(j-1)$ is even, by Invariant $I_3$,  $L_{w_{j-1}}$ lies entirely on $\Pi(\alpha(w_j))$, and the $Y$-line through 
 $init(w_j)$ intersects (the extension of) the horizontal hand of $L_{w_{j-1}}$.  
 We construct a vertical arrow for $L_{w_j}$ that starts at $init(w_j)$ and touches $L_{w_{j-1}}$ (we extend $L_{w_{j-1}}$ if necessary), e.g., see Fig.~\ref{invariants}(d). 
 We then construct the other hand of $L_{w_j}$ using a $(+Z)$-arrow that starts at $init(w_j)$ and stops at
 $\Pi(\alpha(w_{j+1})$, and thus satisfy Invariant $I_3$.
 Note that for the last vertex $v_k$, 
  $L_{v_k}$ either has a peak  on $\Pi(\alpha(v_{k+1}))$ or lies
 entirely on $\Pi(\alpha(v_{k+1}))$ (depending on the parity of $k$).

This completes the construction of $\Psi'_1$, which already satisfies $I_3$. 
 Therefore, it remains to show that  $\Psi'_1$ satisfies $I_1$ and $I_2$. 
 It is straightforward to observe that all the adjacencies have been realized. 
 We thus need to show that we did not create any unnecessary adjacency. 
 The only nontrivial part of the construction is the modification of the 
 blue $L$-shapes to realize the blue-blue adjacencies, and it suffices to 
 show that we do not intersect any unnecessary blue $L$-shape or any
 red $L$-shape during this process.  By construction, the polygonal path determined by the blue $L$-shapes
 is monotonically increasing along the $Z$-axis, and hence 
 the modification does not create any unnecessary blue-blue adjacency. 
 Moreover, by construction, the joint of the red vertices, and thus the initiators
 of the  blue vertices are also in general position. Therefore, the modification does not 
  introduce any unnecessary red-blue adjacency. Hence $\Psi'_1$ satisfies $I_1$. 
  Since the blue vertices of degree one in $\Psi'_1$ are represented as 
  points directly above (with respect to $\Pi_{xy}$) the joints of the red $L$-shapes,
  $\Psi'_1$ satisfies $I_2$.

\subsection{Construction of $\Psi'_i$} 
We now assume that $i>1$ and for every $q<i$, $\Psi'_q$ satisfies the Invariants $I_1$--$I_3$. Here we describe the construction of $\Psi'_i$. 

\textbf{Red-Red and Red-Blue (Equivalently, Blue-Red) Adjacencies:}
 Let $\C_i$ be the red chain $v_j,u_{j+1},\ldots,v_{j+q}$ with head $v_h$ and tail $v_t$. The tail $v_t$ has two red neighbors preceding it on $P$: one  is $v_{j+q}$, and  the other one is its red parent  $v_r$. 
 We distinguish the following two cases. 

\textbf{Case 1 ($v_r$ belongs to $\C_i$):} 
 By Invariant $I_3$ and the choice of $\alpha(\cdot)$ value,
 $L_{v_h}$ either entirely lies on $\Pi(\alpha(v_t))$,
 or contains only a peak on $\Pi(\alpha(v_t))$.  

If $L_{v_h}$ contains only a peak $o$ on $\Pi(\alpha(v_t))$,
 then the idea is to draw  $\C_i$ and $blue(\C_i)$ independently,
 and then merge the drawing such that $L_{v_j}$ touches $o$.
 Specifically,  we find a rectangle  $R$ on $\Pi(\alpha(v_t))$
 with the bottom-left corner at $o$.  We construct a  drawing
 $D$ of $\C_i$ and $blue(\C_i)$ by  mimicking the construction
 of $\Psi'_1$, and place $D$ (possibly by scaling down)  inside
 $R$, e.g., see Fig.~\ref{GammaI}(a). Note that $D$ does not
 contain any blue-blue adjacencies. By construction, one hand
 $r$ of $L_{v_j}$ lies on $R$ (the other is represented by a
 $(+Z)$-arrow). We adjust the placement of $D$ such that the
 peak of $r$ coincides with $o$. We then perturb  $D$ such
 that the initiators of $blue(\C_i)$, and the degree-one blue
 vertices of $\Psi'_{i-1}$ lie in general position.
 
\begin{figure}[pt]
\centering
\includegraphics[width=.9\textwidth]{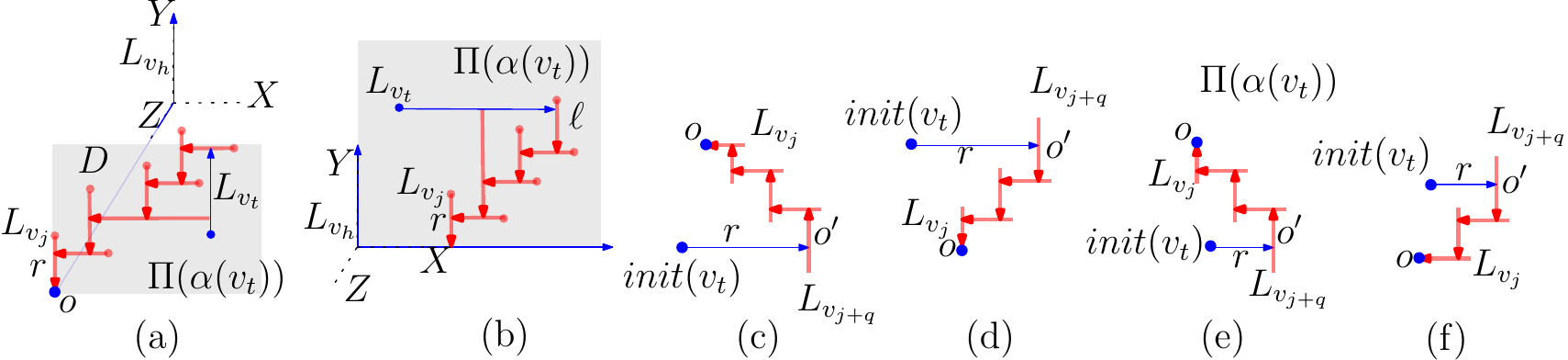}
\caption{(a)--(b) Construction of $\Psi'_i$, where $v_r$ belongs to $\C_i$. (c)--(f) Different scenarios  for the construction of $D$. The $(+Z)$-arrows are not shown for clarity.}
\label{GammaI}
\end{figure}

If $L_{v_h}$ lies entirely on $\Pi(\alpha(v_t))$, then by Invariant $I_3$, $L_{v_h}$ is non-degenerate.
   We find a rectangle  $R$ on $\Pi(\alpha(v_t))$ with one side along the   horizontal 
   hand of $L_{v_h}$. We then construct a drawing 
    $D$ of $C_i$ and $blue(\C_i)$ by  mimicking the construction of $\Psi'_1$.
  Recall that such construction enforces $L_{v_{j+q}}$ to contain
 a $X$-segment. Instead, we use a symmetric construction such that 
 $L_{v_{j}}$  contains a $Y$-segment on $\Pi(\alpha(v_t))$, and thus the hand $\ell$ of $L_{v_{j+q}}$
 that lies on  $\Pi(\alpha(v_t))$ may be horizontal or vertical (depending on the number 
 of vertices in $\C'_i$). If $\ell$ is vertical (resp., horizontal), then we represent $L_(v_t)$
 as a horizontal (resp., vertical) arrow with origin at $init(v_t)$. 
 It is now straightforward to  place 
 $D$ (possibly taking vertical reflection) inside $R$
  such that $L_{v_j}$ touches the horizontal hand of $L_{v_h}$, e.g., see Fig.~\ref{GammaI}(b).

Observe that in both the cases  (above), we have a special scenario, as follows:  
 If $v_r$ coincides with $v_j$, then by the construction of the red
  $L$-shapes $L_{v_r}$ does not contain any $(+Z)$-arrow,
 e.g., see  Fig.~\ref{invariants}(e). This is fine 
  as long as $v_j \not = v_{j+q}$, because $v_r$ already contains three incidences
  at its current hand. If $v_j$ coincides with  $v_{j+q}$, then we create a 
  $(+Z)$-arrow for $L_v$ that stops at $\Pi(\alpha(w))$, where $w$ is a blue child of $v_j$,
  e.g., see  Fig.~\ref{invariants}(f).

\textbf{Case 2 ($v_r$ belongs to $\Psi'_{i-1}$):} In this scenario, the degree of $v_t$  in $G'_{i-1}$ is one,  and by Invariant $I_2$, $v_t$ is  represented as a point in $\Psi'_{i-1}$. 
  We distinguish two subcases depending on the size of $\C_i$.

\textit{Case 2a ($\C_i$ has two or more vertices):}  
 If $L_{v_h}$ contains only a peak $o$ on $\Pi(\alpha(v_t))$, then 
 we represent $L_{v_t}$ using a rightward arrow $r$ that starts at $init(v_t)$
 and stops at some point $o'$ to the right of the $Y$-line though $o$.
 We then  construct a  drawing $D$ of $\C_i$ and $blue(\C_i)$ on $\Pi(\alpha(v_t))$ mimicking 
 the construction of $\Psi'_1$. However, this is simpler since the red parent of 
 $v_t$ does not belong to $\C_i$. We ensure that $L_{v_j}$ has a $Y$-segment. 
 Figs.~\ref{GammaI}(c)--(f) show  all distinct scenarios.


Assume now that $L_{v_h}$ lies  entirely on $\Pi(\alpha(v_t))$.  By Invariant $I_3$ and  the choice of $\alpha(\cdot)$ values, $L_{v_h}$ is non-degenerate and 
 (the extension of) its horizontal hand intersects the  vertical line through $init(v_t)$ in $\Psi'_{i-1}$.
 The drawing in this case is illustrated in Figs.~\ref{GammaI2}(a)--(b). Appendix~\ref{case2b}
 includes the details.

Note that in both cases we may need to perturb the drawing $D$ such that the $L$-shapes in $D$ do not create
 any unnecessary intersections, and $blue(\C'_i)$ and the degree-one blue vertices of $\Psi'_{i-1}$ lie in general position.
 

\begin{figure}[pt]
\centering
\includegraphics[width=.85\textwidth]{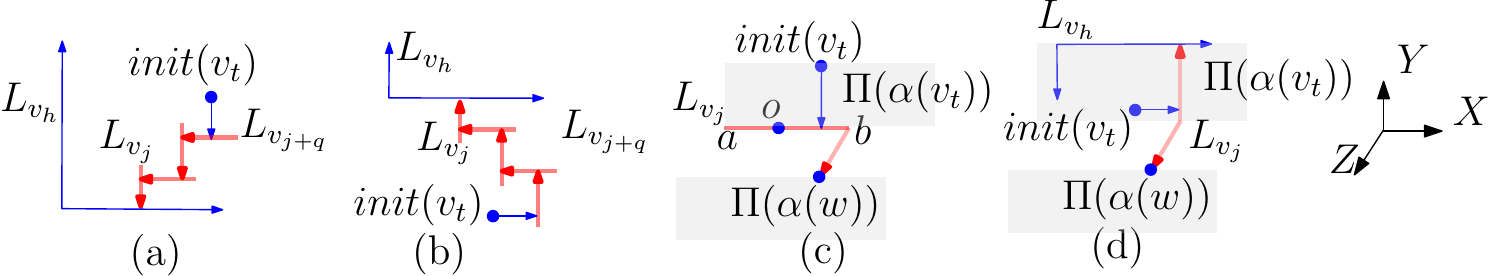}
\caption{(a)--(b) Different scenarios  while constructing $D$. (c)--(d) Case 2b.}
\label{GammaI2}
\end{figure}

\textit{Case 2b ($\C_i$ has only one vertex):} This case is straightforward to process, e.g.,  see Figs.~\ref{GammaI2}(c)--(d). Details are included in Appendix~\ref{case2b}. 

 

\textbf{Blue-Blue Adjacencies:} 
Let $w_1(=v_t),\ldots,w_k$ be the blue chain $\C'_i$. If $\C'_i$ does not include all the blue vertices of $G'$, then let $w_{k+1}$ be the first blue vertex following $w_k$ on $P$. Note that $w_k$ is the head of $\C_{i+1}$, and $w_{k+1}$ is the tail of  $\C_{i+1}$. On the other hand, if  $\C'_i$ contains all the blue vertices of $G'$, then consider a dummy vertex $w_{k+1}$. 

If $k=1$, then all the blue-blue adjacencies in $G'_i$ are present in $G'_{i-1}$, and we only  construct a $(+Z)$-arrow which starts at $init(w_1) = init(v_t)$ and stops at $\Pi(\alpha(w_2))$. Otherwise, we modify $L_{w_j}$, where $1\le j\le k$, to realize the blue-blue adjacencies. By Invariant $I_2$ and the  initial construction of $blue(\C_i)$, all $L_{w_j}$ except $L_{w_1}$ are represented as  distinct points on $\Pi(\alpha(w_j))$, which are in general position. Therefore, we can modify $L_{w_j}$  satisfying Invariant $I_3$ in the same way as we realized the blue-blue adjacencies in $\Psi'_1$.

The argument that $\Psi'_i$ satisfies the induction invariants are similar to that of $\Psi'_1$, but 
 we need to consider also the drawing $\Psi'_{i-1}$. While drawing of $\C'_i$ and 
 $blue(\C_i)$, we ensured the general position property, and thus satisfied Invariant $I_2$. 
 This general position property leads  us to the argument that no unnecessary adjacency is created
 during the modification of the blue $L$-shapes (i.e., Invariant $I_1$).
 Finally, the Invariant $I_3$ follows from the modification of the blue $L$-shapes.


Finally, we modify $L_{v_n}$ to realize the edge $(v_1,v_n)$. Since $v_n$ is blue, by Invariant $I_3$, one of the hands of $L_{v_n}$ can be extended, and we extend this hand using about three more bends to touch $L_{v_1}$. 
  The following theorem summarizes the result of this section. 
\begin{theorem}
\label{3-reg-bh} 
Every  3-regular  Hamiltonian bipartite graph has   a contact representation where all strings but one are  $B_1$-strings. 
\end{theorem}

\section{Lower Bounds}\label{lb}
\begin{theorem}
\label{5-reg}
No 5-regular  graph admits a string contact representation.
\end{theorem}
\begin{proof}[Proof Outline]
Let $G$ be a 5-regular graph, and suppose for a contradiction that $G$ admits a string contact representation $D$. For each edge $(u,w)$,  if the string of $u$ touches the string of $w$, then direct the edge from $u$ to $w$. Note that $G$ has exactly $\frac{5n}{2}$ edges, hence a vertex   with  
 out-degree $\ge 3$. 
\end{proof}

\begin{theorem}
\label{4-reg}
 $K_5$ (a $4$-regular graph) does not have $L$-contact representation.
\end{theorem}
\begin{proof}[Proof Outline]
Suppose for a contradiction that $K_5$ admits an $L$-contact representation, and let $D$ be such a representation of $K_5$. Let $v_1,\ldots, v_5$ be the vertices of $K_5$. Observe that any axis-aligned $L$-shape must entirely lie on one of the three types of plane: $\Pi_{xy}$, $\Pi_{yz}$, and $\Pi_{xz}$. Since there are five $L$-shapes in $D$,  
 the plane types for at least two $L$-shapes must be the same.

Without loss of generality assume that   $L_{v_1}$ and $L_{v_2}$ both lie on $\Pi_{xy}$. Since $v_1$ and $v_2$ are adjacent,  the planes of $L_{v_1}$ and $L_{v_2}$ cannot be distinct.    Therefore, without loss of generality assume that they coincide with $\Pi(0)$. Since $v_i$, where $3\le i\le 5$,  is adjacent to both $v_1$ and $v_2$, $L_{v_i}$ must share a point $a_i$ with   $L_{v_1}$ and a point $b_i$ with $L_{v_2}$. Since no three strings meet at a point in $D$, the points $a_i$ and $b_i$ are distinct. The rest of the proof claims that  the polygonal path  $P_i$ of $L_i$ that starts at $a_i$ and ends at $b_i$, lies entirely on $\Pi(0)$. 
 This property of $D$ can be used to argue that  $D$ is a string contact representation of $K_5$ on $\Pi(0)$, which contradicts  that  $K_5$ is a non-planar graph. Appendix~\ref{applb} includes the details.
\end{proof}

\begin{theorem}
\label{3-reg}
 $K_{3,3}$ (a $3$-regular graph) does not have segment contact layout.
\end{theorem}
\begin{proof}[Proof Outline]
The proof is based on the observation that any contact representation of a 4-cycle, i.e., a cycle of four vertices, with axis-aligned $B_0$ strings, lies entirely on a single plane. Furthermore, two adjacent segments completely determine this plane. Since the vertices of $K_{3,3}$ can be covered by two 4-cycles that share an edge,  any string contact representation of $K_{3,3}$ must lie  on a single plane.  A detailed proof is in Appendix~\ref{applb}.
\end{proof}   

\section{Directions for Future Research}

Improving the complexity bound of the string contact representations for the graph classes we discussed in Theorems~\ref{4-reg-gu}--\ref{3-reg-b} is a natural avenue to explore.  But the most fascinating question is   whether every $3$-regular graph admits an $L$-contact representation in $\mathbb{R}^3$, even with the `triangle-free' constraint. 
 
\bigskip
\noindent{\bf Acknowledgments.} 
 The author is thankful to Anna Lubiw and anonymous reviewers for their detailed comments to improve the presentation of the paper.


\bibliography{ref}

\newpage
\appendix

\section{Representations with String Complexity 4}
\label{sec:4-string}
\begin{theorem}
\label{4-reg-u} 
Every graph with maximum degree $4$ admits a  string contact representation with   complexity $4$. 
\end{theorem}
\begin{proof}
The proof is based on the concept of geometric thickness.  Duncan et al.~\cite{DuncanEK04} proved that every graph with maximum degree 4 has geometric thickness two, and if the edges are allowed to be orthogonal, then such a drawing $\mathcal{D}$  can be computed satisfying the following properties. 

\begin{enumerate}
\item[A.] Every vertex in $\mathcal{D}$  has unique $x$ and $y$-coordinates, and each edge $e$ in $\mathcal{D}$  is drawn as a sequence of two axis-aligned line segments between the end vertices of $e$. 

\item[B.] Each planar layer in $\mathcal{D}$  consists of paths and cycles. Each path 
 or cycle $v_1,\ldots,v_k$ in the first (resp., second) layer, is drawn inside a vertical (horizontal) slab, where the path  $v_1,\ldots,v_k$ is drawn as an $x$-monotone ($y$-monotone) polygonal path. 
\end{enumerate}
Fig.~\ref{4-string}(a) illustrates such a drawing $\mathcal{D}$, the edges of one planar layer are drawn using thin lines, and the other planar layer is drawn using thick lines.

\begin{figure}[pt]
\centering
\includegraphics[width=.8\textwidth]{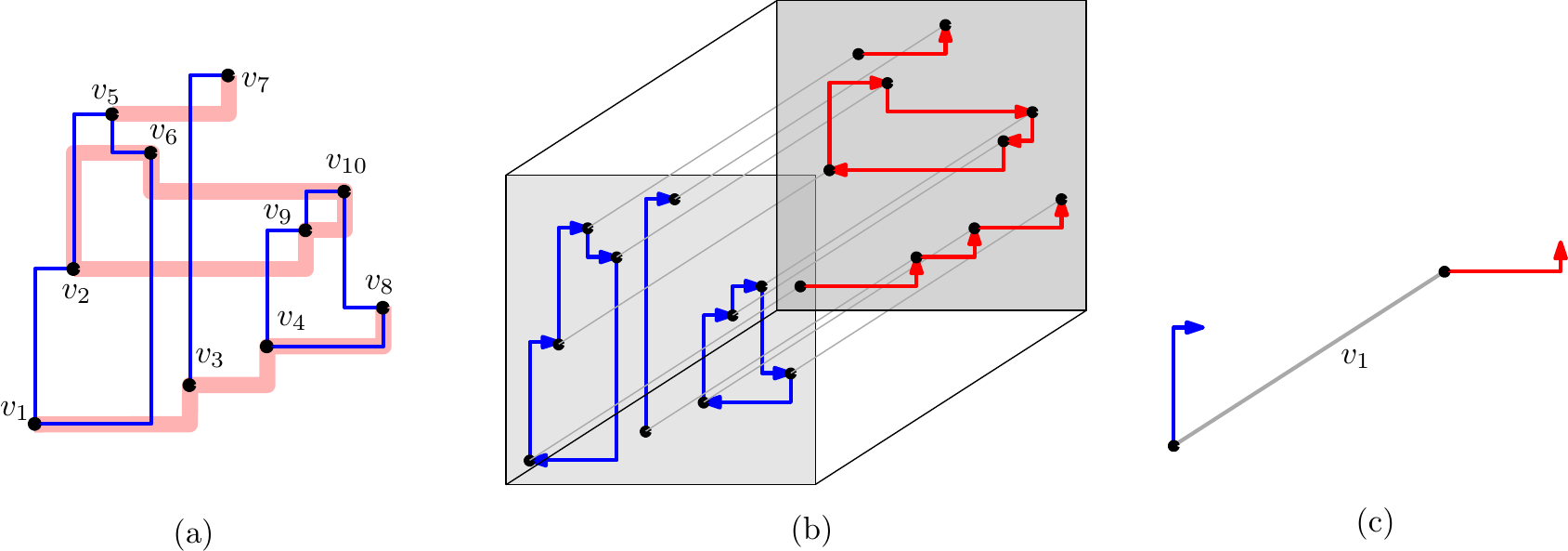}
\caption{(a) Illustration for $\mathcal{D}$. (b) Transformation into a string contact representation. (c) The string corresponding to vertex $v_1$. }
\label{4-string}
\end{figure}

For each cycle $C$ in the first (second) layer, we direct the edges on $C$ in clockwise order, and for each path $P$, we direct the edges of $P$ from left to right (resp., bottom to top). Consequently, each vertex now has out-degree at most   one in each layer. 
 We lift the edges on the second layer up by one unit,   representing each vertex using a unit $Z$-line.  Fig.~\ref{4-string}(b) illustrates a schematic representation of the resulting drawing.   This yields a contact representation of $G$ using $B_4$-strings, where the string of each vertex consists of its outgoing edges and the $Z$-line that connects these outgoing edges.  Fig.~\ref{4-string}(c) illustrates such a $B_4$-string.
\end{proof}

\smallskip
\noindent\textbf{Theorem~\ref{4-reg-gu}.} \emph{
Every  triangle-free Hamiltonian graph $G$ with maximum degree four has a contact representation where all strings but one are  $B_3$-strings. 
}
\begin{proof}
Let $C=(v_1,\ldots,v_n)$ be a Hamiltonian cycle of $G$. Let $H$ be the graph obtained after removing the  Hamiltonian edges from $G$.
 Since every vertex of $H$ is of degree at most two, $H$ is a union of vertex disjoint cycles and paths. We transform each path $P=(w_1\ldots,w_k)$ of $H$ into a cycle by adding a subpath of one or two dummy vertices between
 ($w_1$ and $w_k$) depending on whether $P$ has one or more vertices. 


 Let $Q_1,\ldots, Q_k$ be the cycles in $H$. For each cycle $Q_i$, where $1\le i\le k$, we construct a staircase representation $\Psi'_i$ of $Q_i$ on $\Pi(0)$. If $Q_i$ is a cycle with odd number of vertices, then we construct the staircase representation such that the leftmost segment among the topmost horizontal segments corresponds to the vertex with the lowest index in $Q_i$. For example, see the topmost staircase of Fig.~\ref{3-string}(a). We then place  the staircase representations diagonally along a line with slope $+1$. We ensure that the horizontal and vertical slabs containing $\Psi_i$ do not intersect $\Psi_j$, where $1\le j(\not=i)\le k$. We  refer to this representation as $\Psi_H$.  
 
Consider now the edges of the Hamiltonian cycle $C=(v_1,\ldots,v_n)$. Note that each vertex $v_j$, where $1\le j\le n$, is represented using an axis-aligned arrow $r_j$ in $\Psi_H$. 
 For each arrow $r_j$,  we construct a $(+Z)$-arrow $r'_j$ of length $j$ that starts at the origin of $r_j$. Consequently, the plane  $\Pi(j)$ intersects only those arrows $r'_q$, where $j\le q\le n$. Let $I_j$  be the set of intersection points on $\Pi(j)$. By construction $\Psi_H$ satisfies the following sparseness property: 

\begin{description}
\item[Sparseness of $\Psi_H$:]
 Any vertical (resp., horizontal) line on $\Pi(j)$ contains at most one point (resp., two points) from $I_j$. For every pair of points $p,q$ that belong to $I_j$ and lie on the same horizontal line, the corresponding vertices are adjacent in $H$, and belong to a distinct cycle with odd number of vertices in $H$.  
\end{description}

\begin{figure}[h]
\centering
\includegraphics[width=\textwidth]{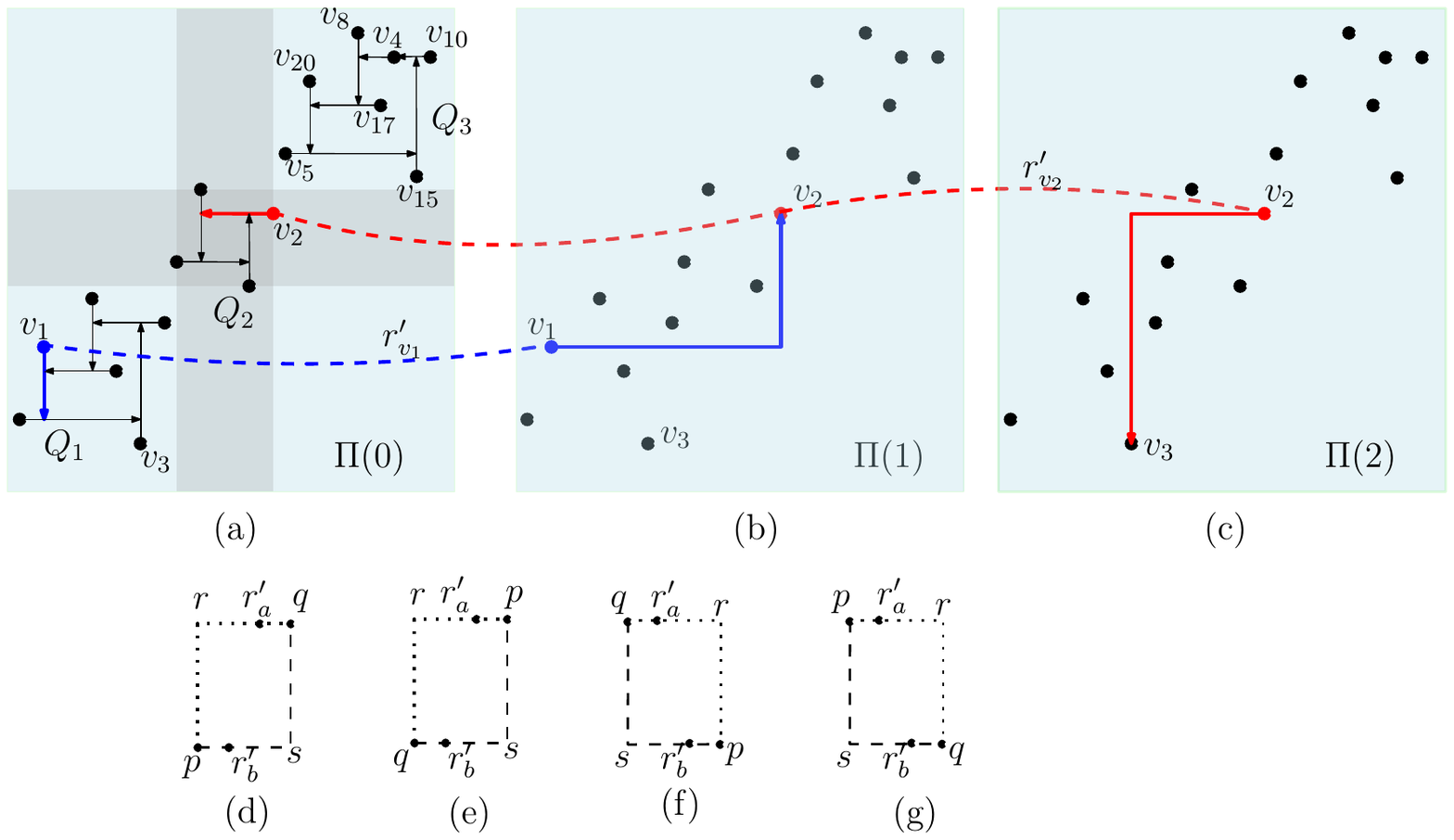}
\caption{ Illustration for the proof of Theorem~\ref{4-reg-gu}: (a) $\Psi_H$, (b)--(c) Extensions of $r'_{v_1}$ and $r'_{v_2}$, where the dashed lines are $(+Z)$-lines. (d)--(g) Illustration for Cases 1--3. 
}
\label{3-string}
\end{figure}

For each $j$ from $1$ to $n-1$, we realize the adjacency between    $v_j$ and $v_{j+1}$ by extending $r'_j$ on $\Pi(j)$. Note that it suffices to use two bends to route $r'_j$ to touch $r'_{j+1}$, where one bend is to enter $\Pi(j)$ and the other is to reach $r'_{j+1}$.  Figs.~\ref{3-string}(b)--(c) illustrate  the extension of $r'_j$.  We now claim that one can find such an extension of $r'_j$ without introducing any crossing.  Assume that $r'_j$ and $r'_{j+1}$ intersect $\Pi(j)$ at points $p$ and $q$, respectively, and  suppose for a contradiction that any  2-bend extension of $r'_j$ to touch $r'_{j+1}$  on $\Pi(j)$ would  introduce an unnecessary adjacency.  We  now consider the following scenarios. 

\textbf{Case 1 ($p$ lies below and to the left of $q$):}  
We refer to the configuration of Fig.~\ref{3-string}(d). Let $R_{pq}$ be the rectangle determined by $p$ and $q$ on $\Pi(j)$. Let $r$ and $s$ be the top-left and bottom-right corners of $R_{pq}$. Assume that both $p,r,q$ and $p,s,q$   introduce unnecessary adjacencies, e.g., $p,r,q$ intersects some arrow $r'_a$ and $p,s,q$ intersects some arrow $r'_b$.  

By the sparseness property of $\Psi_H$, the arrows $r'_a$ and $r'_b$ cannot lie on $pr$ or $qs$, and hence, they   
 must intersect the segments $rq$ and $ps$, respectively. Since the intersection point with $r'_a$ lies to the left of $q$, we have $a<j+1$. Moreover, since $v_a$ and $v_j$ are distinct vertices, we have $a<j$.  Consequently, $v_a$ cannot intersect $\Pi(j)$, and we can extend $r'_j$ along $p,r,q$.

\textbf{Case 2 ($p$ lies above and to the right of $q$):} 
 This scenario is similar to Case 1, e.g., see Fig.~\ref{3-string}(e). Since the intersection point of $r'_a$ is to the left of $p$, $a<j$. Consequently, $v_a$ cannot intersect $\Pi(j)$, and we can extend $r'_j$ along $p,r,q$.

\textbf{Case 3 (Otherwise):} Since $(v_j,v_{j+1})$ is a Hamiltonian edge, by the sparseness property of $\Psi_H$, $p$ and $q$   cannot lie on the same horizontal line on $\Pi(j)$. The remaining cases are as follows: (I) $p$ lies below  and to the right of $q$, and (II) $p$ lies above  and to the left of $q$. Figs.~\ref{3-string}(f)--(g) illustrate these two cases. 
 By the sparseness property, $\{v_{j+1},v_a\}$ and $\{v_{j},v_b\}$ correspond to distinct cycles in $H$. Since the cycles of $H$ are placed diagonally along a line with slope $+1$, none of these two configurations can arise.

Finally, it is straightforward to realize $(v_1,v_n)$ by routing $r'_n$ on $\Pi(n)$ using two bends, and then moving downward to touch $r_1$. Therefore, the string representing $v_n$  is a $B_4$-string.
\end{proof}

\section{Details of Section~\ref{pre}}
\label{defn}

\textbf{Staircase representation:} Consider first the case when $k$ is even. We first draw an $xy$-monotone orthogonal polyline  $\mathcal{O}$ with $k-2$ unit-length segments $l_2,\ldots,l_{k-1}$, where the segments $l_2,l_4,\ldots$ are vertical, and $l_3,l_5,\ldots$ are horizontal. We then join the  end points of $\mathcal{O}$ using a horizontal line segment $l_1$ and a vertical  line segment $l_k$, as shown in Fig.~\ref{even2}(a).  We order the edges of the resulting orthogonal polygon $O$ in counterclockwise order, and assign $w_i$ the segment $l_i$, where $1\le i\le k$. We then extend the horizontal segments, except $l_1$,  one-half unit to the  right, and the vertical segment, except $l_k$, one-half unit upward. Finally, we extend the segments corresponding to $l_1$ and $l_k$ one-half unit to the left and downward, respectively. Observe  that the extended segments do not introduce any crossing.  Consequently, each vertex $w_i$ can now be represented as an  axis-aligned arrow $r_i$, where the extended end of the segments   correspond to the origins, e.g., see  Fig.~\ref{even2}(b). Furthermore, the origins of $r_1,\ldots,r_k$ are in \emph{general position}, i.e., no two of them  have the same $x$ or $y$-coordinate. 

\begin{figure}[h]
\centering
\includegraphics[width=.7\textwidth]{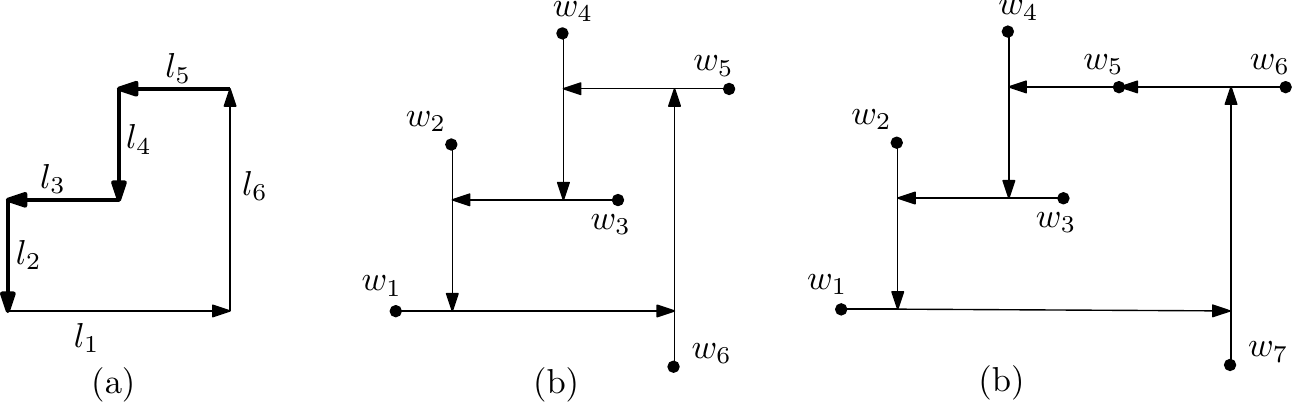}
\caption{Construction of a staircase representation. }
\label{even2}
\end{figure}

If $k$ is odd, then we take a staircase representation of a cycle of $k-1$ vertices, and then subdivide the topmost horizontal segment to create the a new arrow, as illustrated in Fig.~\ref{even2}(c).

\section{Details of Section~\ref{string-contact}}
\label{case2b}

\textit{Case 2a ($\C_i$ has two or more vertices):}   

 If $L_{v_h}$ contains only a peak $o$ on $\Pi(\alpha(v_t))$, then 
 we represent $L_{v_t}$ using a rightward arrow $r$ that starts at $init(v_t)$
 and stops at some point $o'$ to the right of the $Y$-line though $o$.
 We then  construct a  drawing $D$ of $\C_i$ and $blue(\C_i)$ on $\Pi(\alpha(v_t))$ mimicking 
 the construction of $\Psi'_1$. However, this is simpler since the red parent of 
 $v_t$ does not belong to $\C_i$. We ensure that $L_{v_j}$ has a $Y$-segment, and thus   
 the hand of $L_{v_{j+q}}$  on $\Pi(\alpha(v_t))$ may be horizontal or vertical.
 It is now straightforward to place $D$ (possibly taking vertical reflection)  
  such that $L_{v_j}$ touches $L_{v_h}$ at $o$ and $L_{v_t}$ touches $L_{v_{j+q}}$
 at $o'$. Figs.~\ref{GammaI}(c)--(f) show  all distinct scenarios.

 Assume now that $L_{v_h}$ lies  entirely on $\Pi(\alpha(v_t))$.
 By Invariant $I_3$ and  the choice of $\alpha(\cdot)$ values, $L_{v_h}$ is non-degenerate and 
 (the extension of) its horizontal hand intersects the  vertical line through $init(v_t)$ in $\Psi'_{i-1}$.
 We now  construct a  drawing $D$ of $\C_i$ and $blue(\C_i)$ on $\Pi(\alpha(v_t))$ mimicking 
 the construction of $\Psi'_1$, but ensuring that  
 $L_{v_{j}}$  contains a $Y$-segment on $\Pi(\alpha(v_t))$. Consequently, the hand $\ell$ of $L_{v_{j+q}}$
 that lies on  $\Pi(\alpha(v_t))$ may be horizontal or vertical (depending on the number 
 of vertices in $\C'_i$). If $\ell$ is vertical (resp., horizontal), then we represent $L_{v_t}$
 as a horizontal (resp., vertical) arrow with origin at $init(v_t)$. 
 It is now straightforward to  place 
 $D$ (possibly taking vertical reflection) on  $\Pi(\alpha(v_t))$ 
  such that $L_{v_j}$ touches the horizontal hand of $L_{v_h}$, and   
 $L_{v_t}$ touches $L_{v_{j+q}}$ at $o'$,
 e.g., see Figs.~\ref{GammaI2}(a)--(b).

Note that in both cases we may need to perturb the drawing $D$ such that the $L$-shapes in $D$ do not create
 any unnecessary intersections, and $blue(\C'_i)$ and the degree-one blue vertices of $\Psi'_{i-1}$ lie in general position.

\textit{Case 2b ($\C_i$ has only one vertex):}  
If $L_{v_h}$ contains only a peak $o$ on $\Pi(\alpha(v_t))$, then we construct $L_{v_j}$ as a horizontal line segment $ab$ that passes through $o$, and represent $L_{v_t}$ as a vertical arrow that touches $L_{v_j}$, e.g., see Fig.~\ref{GammaI2}(c). We then construct another hand of $L_{v_j}$ using a $(+Z)$-arrow $r$ that starts at $b$, and ends at $\Pi(\alpha(w))$, where $w$ is the blue child of $v_j$. We create the initiator of $w$ at the peak of $r$.

Assume now that $L_{v_h}$ lies  entirely on $\Pi(\alpha(v_t))$. By Invariant $I_3$, the $Y$-line through $init(v_t)$ intersects (the extension of) the horizontal hand of $L_{v_h}$. It is thus straightforward to construct $L_{v_j}$
 as a $Y$-segment $ab$ that touches the horizontal hand of $L_{v_h}$ at $a$, and then construct $L_{v_t}$ as a horizontal arrow with origin $init(v_t)$ that touches $L_{v_j}$.  We then construct another hand of $L_{v_j}$ using a $(+Z)$-arrow $r$ that starts at $b$, and ends at $\Pi(\alpha(w))$, where $w$ is the blue child of $v_j$. We create the initiator of $w$ at the peak of $r$, e.g., see Fig.~\ref{GammaI2}(d).
 
In both cases we choose $b$ carefully to  ensure the general position property of the initiators.

\section{Details of Section~\ref{lb}}
\label{applb}
 
\subsubsection{Proof of Theorem~\ref{5-reg}: }

Let $G$ be a 5-regular graph, and suppose for a contradiction that $G$ admits a string contact representation $D$. For each edge $(u,w)$,  if the string of $u$ touches the string of $w$, then direct the edge from $u$ to $w$. Note that $G$ has exactly $5n/2$ edges. Since each edge is either unidirected  or bidirected, the sum of all out-degrees is at least  $5n/2$. Therefore, there exists a vertex $v$  with  
 out-degree 3 or more. However, by definition, 
 no three strings in $D$ can meet at a point.   Therefore, the out-degree of $v$ cannot be larger than two, a contradiction.
 
\vspace{-.2cm}
\subsubsection{Proof of Theorem~\ref{4-reg}: }
Suppose for a contradiction that $K_5$ admits an $L$-contact representation, and let $D$ be such a representation of $K_5$. Let $v_1,\ldots, v_5$ be the vertices of $K_5$. Observe that any axis-aligned $L$-shape must entirely lie on one of the three types of plane: $\Pi_{xy}$, $\Pi_{yz}$, and $\Pi_{xz}$. Since there are five $L$-shapes in $D$, by  pigeonhole principle, the plane types for at least two $L$-shapes must be the same.  

Without loss of generality assume that   $L_{v_1}$ and $L_{v_2}$ both lie on $\Pi_{xy}$. Since $v_1$ and $v_2$ are adjacent,  the planes of $L_{v_1}$ and $L_{v_2}$ cannot be distinct.  Therefore, without loss of generality assume that they coincide with $\Pi(0)$. Since $v_i$, where $3\le i\le 5$,  is adjacent to both $v_1$ and $v_2$, $L_{v_i}$ must share a point $a_i$ with   $L_{v_1}$ and a point $b_i$ with $L_{v_2}$. Since no three strings meet at a point in $D$, the points $a_i$ and $b_i$ are distinct. The rest of the proof claims that  the polygonal path $P_i$ of $L_i$ that starts at $a_i$ and ends at $b_i$, lies entirely on $\Pi(0)$, and the common point of $L_{v_i}$ and $L_{v_j}$, where $3\le i<j\le 5$, lies on $\Pi(0)$. These properties can be used to argue that  $D$ is a string contact representation of $K_5$ on $\Pi(0)$, which contradicts  that  $K_5$ is a non-planar graph.

 We now claim that the polygonal path $P_i$ of $L_i$ that starts at $a_i$ and ends at $b_i$, lies entirely on $\Pi(0)$. Since $a_i$ and $b_i$ both lie on $\Pi(0)$, the claim is straightforward to verify when $P_i$ is a straight line segment. Therefore, assume that $P_i$ contains the joint $o_i$ of $L_i$. In this scenario, both the segments $a_io_i$ and $b_io_i$ are perpendicular to $\Pi(0)$. Since $P_i$ does not contain any line segment other than  $a_io_i$ and $b_io_i$, $a_i$ must coincide with $b_i$, a contradiction.
 
Observe now that at least one hand of $L_{v_i}$ lies on $\Pi(0)$. Therefore, the other hand of $L_{v_i}$ lies either  on $\Pi(0)$ or perpendicular to $\Pi(0)$. Therefore, the common point of $L_{v_i}$ and $L_{v_j}$, where $3\le i<j\le 5$, must lie on $\Pi(0)$. Consequently, $D$ is a string contact representation of $K_5$ on $\Pi(0)$, which contradicts  that  $K_5$ is a non-planar graph.

\vspace{-.2cm}
\subsubsection{Proof of Theorem~\ref{3-reg}: }
Suppose for a contradiction that $K_{3,3}$ admits a  segment contact representation, and let $D$ be such a representation of $K_{3,3}$. Let $\{v_1,v_2,v_3\}$ and $\{w_1,w_2,w_3\}$ be the two vertex sets corresponding to $K_{3,3}$. Observe now that any axis-aligned polygon of four line segments must lie on one of the following three types of plane: $\Pi_{xy}$, $\Pi_{yz}$, and $\Pi_{xz}$. Therefore, without loss of generality we may assume that segments corresponding to the cycle $v_1,w_1,v_2,w_2$ lie entirely on $\Pi(0)$.

Since the segments corresponding to $v_1,w_1,v_2,w_2$ bounds a non-degenerate region of $\Pi(0)$, the segments corresponding to $u_1,w_1$ cannot be collinear, and hence they would determine the plane $\Pi(0)$. Consequently, the cycles $u_1,w_1,u_3,w_3$ would force the segments of  $u_3$ and $w_3$ to lie on $\Pi(0)$. Consequently, $D$ must be a string contact representation of $K_{3,3}$ on $\Pi(0)$, which contradicts  that $K_{3,3}$ is a non-planar graph.

\begin{figure}[pb]
\centering
\includegraphics[width=.5\textwidth]{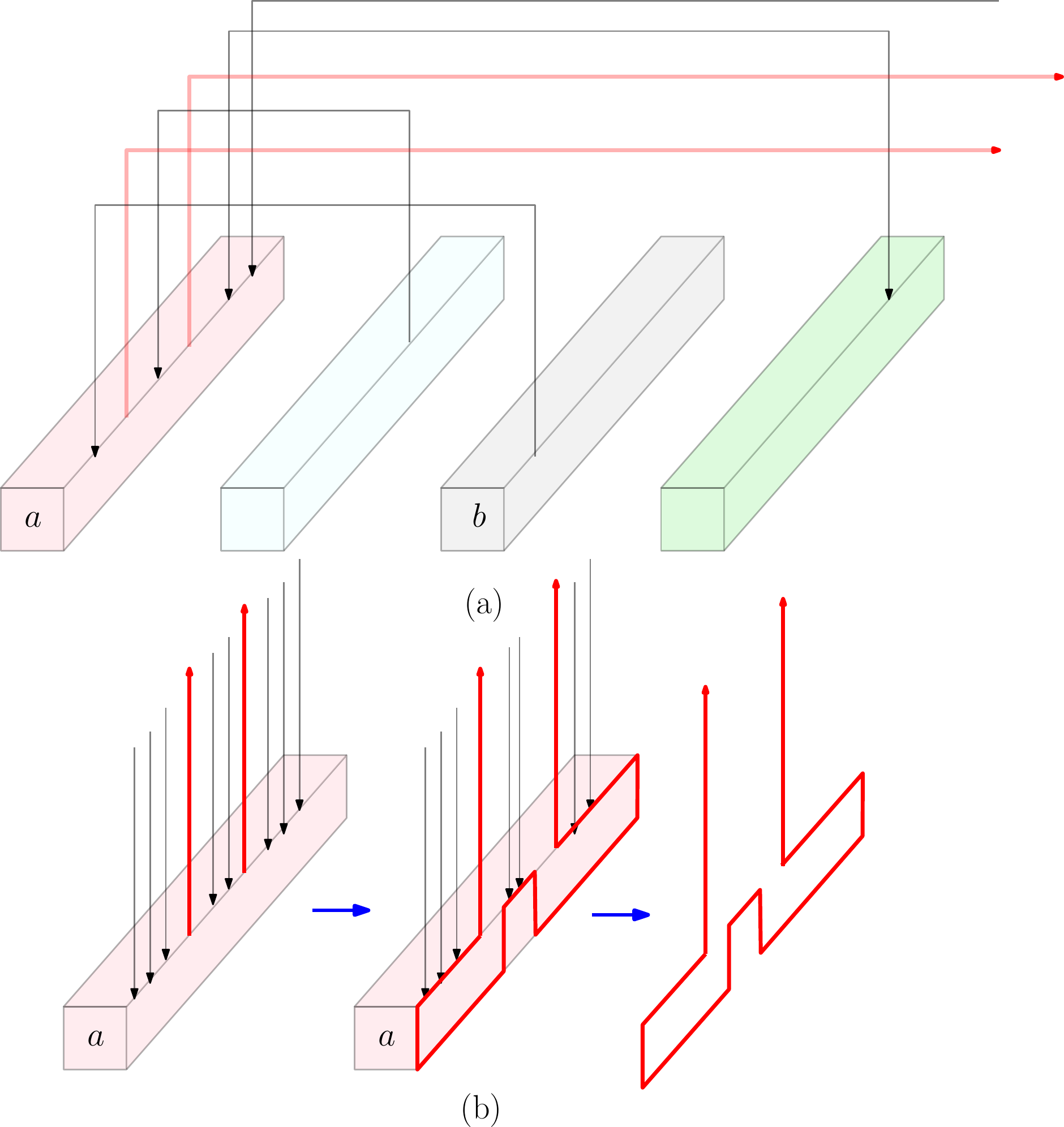}
\caption{Transforming a graph orientation into a string contact representation.}
\label{general}
\end{figure}

\section{From Graph Orientation to String Contact Representations}
\label{app:general}
Given an edge oriented graph (each edge is unidirected), where every vertex 
 has outdegree at most two, we can transform it into a string contact representation 
 using constant number of bends, as follows: 
 
Represent vertices parallel boxes as illustrated in Fig.~\ref{general}(a).
 For each edge $(a,b)$ draw a polygonal path between the corresponding boxes $a$ and $b$
  by following the directions Up, Right, Down, and 
  ensure that the edges lie on distinct planes parallel to  $\Pi_{xy}$. 
  The general setup for each box is illustrated in  Fig.~\ref{general}(a).
  One can now construct the string corresponding to a box by connecting 
  the outgoing edges by a polygonal path of constant number of bends, as 
  shown in  Figs.~\ref{general}(b).

\section{A Walkthrough Example}
\label{app:last}

Figs.~\ref{test}--\ref{summary} illustrate a walkthrough example according to the incremental 
 construction described in Section~\ref{sec:lcontact}.
\begin{figure}[h]
\centering
\includegraphics[width=.7\textwidth]{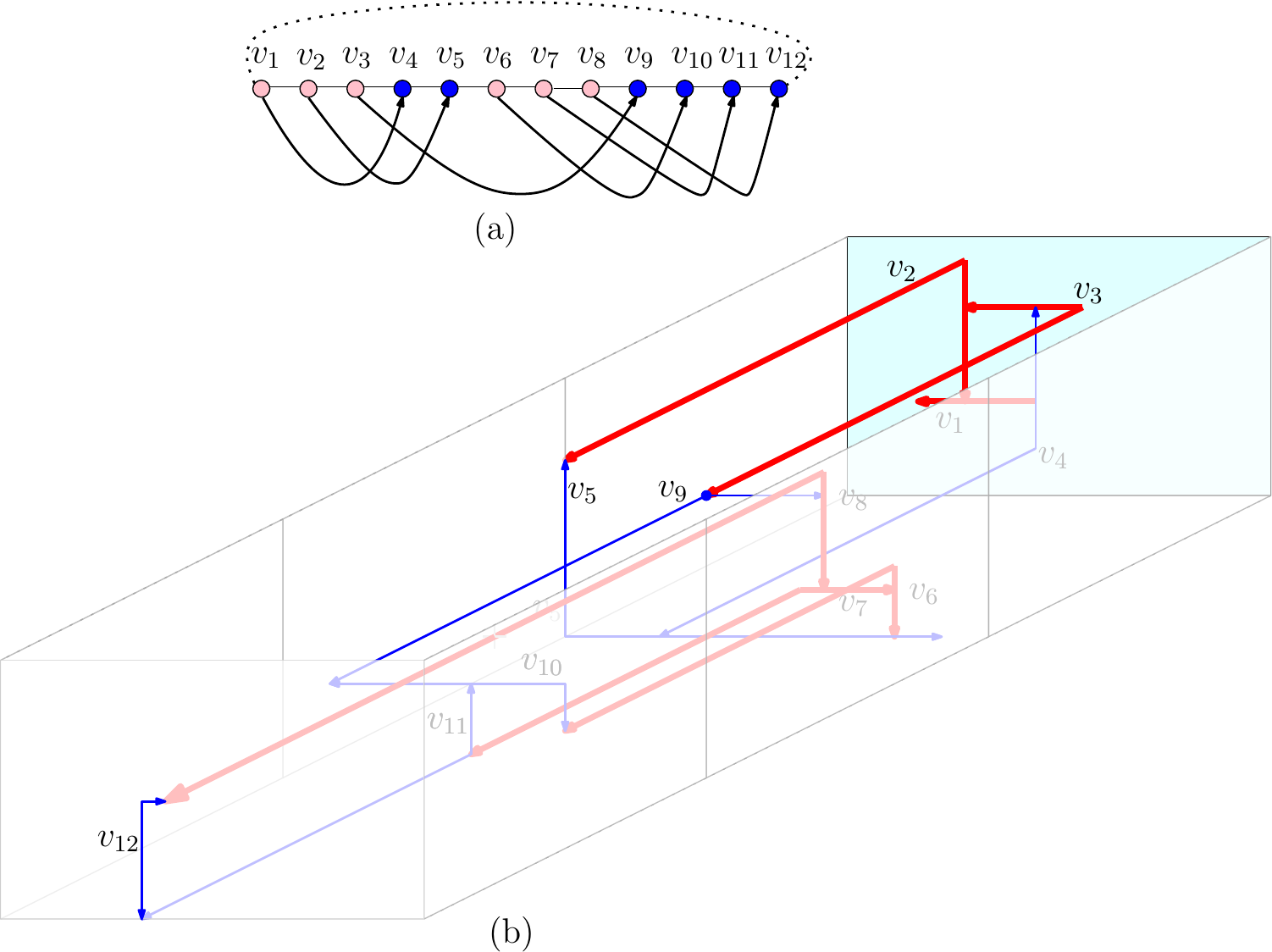}
\caption{(a) A 3-regular Hamiltonian bipartite graph $G$. (b) A schematic representation of an $L$-contact representation of $G$ minus one edge (computed by our algorithm). }
\label{test}
\end{figure}

\begin{figure}[h]
\centering
\includegraphics[width=\textwidth]{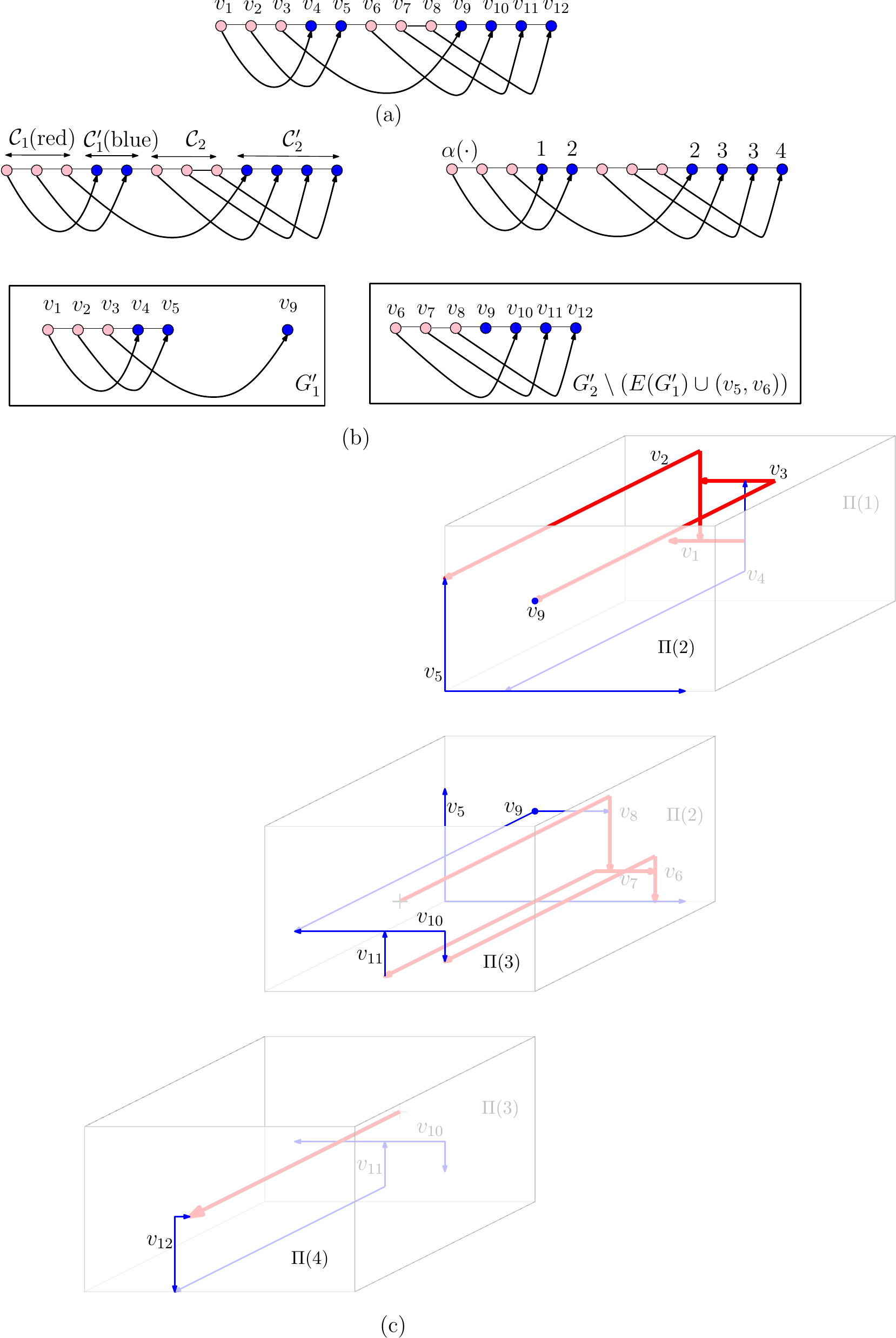}
\caption{(a) A 3-regular Hamiltonian bipartite graph $G$ minus one Hamiltonian edge. 
 (b) Preliminary setup. (c) Incremental Construction. }
\label{summary}
\end{figure}

\end{document}